\documentclass{llncs}

\usepackage{geometry}
\usepackage{epsfig}
\usepackage{tikz}
\usepackage{subfigure}
\usepackage{inputenc}
\usepackage{amssymb} 
\usepackage{amsmath}
\usepackage{array}
\usepackage{listing}
\usepackage{underscore}
\usepackage{graphicx}
\usepackage{eso-pic}
\usepackage{fancyhdr}
\usepackage{algorithmic}
\usepackage{algorithm}
\usepackage{hyperref}  

\usepackage{float}

\usepackage{pdfpages}
 \usepackage{hyperref}

\usepackage{algorithm,algorithmic,verbatim}
\usetikzlibrary{shapes,snakes,positioning,calc}
\usetikzlibrary{arrows,shadows}
\usepackage{hyperref}
\usepackage[
	lambda,
	operators,
        general,
	advantage, 
	sets,     
	adversary,
	landau,
	probability,
	notions,	
	logic,
	ff,
	mm, 
	primitives,
	events,
	complexity,
	asymptotics,
	keys]{sharedcryptocode}

\usepackage{csquotes}

\usepackage{dashbox}
\usepackage{todonotes}

\usepackage{url}

\usetikzlibrary{shapes.callouts}

\usepackage{trace}

\usepackage{mdframed}

\newtheorem{myclaim}[theorem]{Claim}

\newcommand{\bool}{\{0,1\}}
\newcommand{\booln}{\bool^n}

\def\cala{{\cal A}}

\def\calf{{\cal F}}

\def\calx{{\cal X}}
\def\caly{{\cal Y}}

\newcommand{\twopartdef}[4]
{
	\left\{
		\begin{array}{ll}
			#1 & \mbox{if } #2 \\
			#3 & \mbox{if } #4
		\end{array}
	\right.
}

\newcommand{\algo}[1]{\mbox{\tt #1}}
\newcommand{\alA}{\algo{A}}

\newcommand{\Ex}{\mathbb{E}}

\newcommand{\IP}{\algo{IP}}
\newcommand{\RP}{\algo{RP}}

\newcommand{\smallip}{\algo{ip}}

\newcommand{\CMC}{\texttt{CMC}}

\newcommand{\ipa}{\algo{ip}}
\newcommand{\flab}{\algo{flab}}
\newcommand{\vlab}{\algo{vlab}}

\title{Bandwidth-Hard Functions from Random Permutations}
\author{}
\institute{}
\author{Rishiraj Bhattacharyya\inst{1} \and Avradip Mandal\inst{2}}
\institute{University of Birmingham, UK \email{rishiraj.bhattacharyya@gmail.com} \and Skyflow, \email{avradip@gmail.com}}

\begin{document}
\maketitle

\begin{abstract}
    ASIC hash engines are specifically optimized for parallel computations of cryptographic hashes and thus a natural environment for mounting brute-force attacks on hash functions. Two fundamental advantages of ASICs over general purpose computers are the area advantage and the energy efficiency. The memory-hard functions approach the problem by reducing the area advantage of ASICs compared to general-purpose computers. Traditionally, memory-hard functions have been analyzed in the (parallel) random oracle model. However, as the memory-hard security game is multi-stage, indifferentiability does not apply and instantiating the random oracle becomes a non-trivial problem. Chen and Tessaro (CRYPTO 2019) considered this issue and showed how random oracles should be instantiated in the context of memory-hard functions.

  The Bandwidth-Hard functions, introduced by Ren and Devadas (TCC 2017), aim to provide ASIC resistance by reducing the energy advantage of ASICs. In particular, bandwidth-hard functions provide ASIC resistance by guaranteeing high run time energy cost if the available cache is not large enough. Previously, bandwidth-hard functions have been analyzed in the parallel random oracle model. In this work, we show how those random oracles can be instantiated using random permutations in the context of bandwidth-hard functions. Our results are generic and valid for any hard-to-pebble graphs.

\end{abstract}

\section{Introduction}
\label{sec:introduction}
Cryptographic hash functions are ubiquitous in modern-day protocols. They are particularly important in password hashing, and proof of work (POW) based blockchain protocols. However, with recent advances in ASIC (Application Specific Integrated Circuits) hash engines, one can compute standard hash functions more efficiently using parallelization. Brute forcing a password database seems quite feasible in the ASIC environment. This observation led to a growing interest in ASIC resistant hash function design, where the goal is to design a hash function whose \emph{evaluation cost} remains nearly identical, irrespective of the hardware it is being evaluated on.

\noindent\textsc{Memory-Hard Functions.} Memory-Hard functions, proposed by Percival \cite{percival2009stronger}, consider the memory cost to be the main hardware cost that balances the situation across different platforms. In particular, the area advantage of ASICs due to dedicated, small-foot-print hash computation units gets nullified if significant spending of area for memory is required. A function is called memory-hard if it requires a lot of memory capacity to evaluate, even using parallel computation. Thus for a memory-hard function, faster evaluation through ASIC would imply more cost due to memory. There has been a long line of research \cite{percival2009stronger,STOC:AlwSer15,AC:BonCorSch16,iacr/ForlerLW13} on ASIC resistant Memory Hard Functions (MHFs).  Scrypt \cite{percival2009stronger} is the first provably secure (sequentially) memory-hard function. However, it had a data-dependent memory access pattern. Balloon hash by Boneh et. al. \cite{AC:BonCorSch16} was one of the first practical, provably secure memory hard hash functions in the random oracle model with data-independent memory access pattern. However, they did not consider the Parallel Random Oracle Model \cite{STOC:AlwSer15} which models adversarial capabilities more realistically. Both \cite{AC:BonCorSch16} and \cite{STOC:AlwSer15} followed the framework of Dwork, Naor and Wee \cite{C:DwoNaoWee05} which has been applied previously in numerous cryptographic work \cite{EC:ACKKPT16,C:DFKP15,TCC:DziKazWic11,AC:ForLucWen14} and related hardness of pebble games \cite{hewitt1970record,cook1973observation}, hardness of certain computations in the random oracle model \cite{CCS:BelRog93}.

\noindent\textsc{Bandwidth-Hard Functions.}
Recently Ren and Devadas \cite{TCC:RenDev17} proposed the notion of \emph{bandwidth-hardness}. Complementing the principle of memory-bound functions \cite{ndss/AbadiBW03}, they argued faster computation in ASIC is not completely free; the energy spent for on-chip computation and memory accesses is the running cost of the ASIC environment. The energy spent on memory accesses is the same across different platforms, thus a good target for normalizing costs across different platforms. A function is called bandwidth-hard if the energy spent on memory access dominates the evaluation's energy cost.

 Ren and Devadas, however, only considered the sequential model of computation, not the more realistic parallel random oracle model. Blocki et. al. \cite{CCS:BloRenZho18} addressed this and extended the notion of bandwidth-hard functions in the parallel random oracle model. They showed various memory-hard functions like scrypt \cite{percival2009stronger}, Argon2i \cite{biryukov2016argon2}, aATSample and DRSample \cite{CCS:AlwBloHar17} are in fact bandwidth-hard under appropriate cache size. 

 \noindent\textsc{Graph Labelling and Pebbling Complexity.} Memory Hard functions and Bandwidth-Hard are constructed based on \emph{Directed Acyclic Graphs} (DAG). Let $G=(V,E)$ be a DAG over $n$ vertices $V=\{v_1,\cdots,v_n\}$. The vertices of $V$ are sorted in some topological order, a path from $v_i$ to $v_j$ implies $v_i\leq v_j$. $v_1$ denotes the unique source, and $v_n$ denotes the unique sink of the graph. The function is defined by the graph $G$ and a labelling function based on a random oracle $H$. The evaluation of the function on input $x$ is the label $\ell$ of the sink $v_n$, where the label of a node $v_i$ is recursively defined as $\ell_i=H(i,\ell_{u_1},\ell_{u_2},\ldots,\ell_{u_\delta})$ where $u_1,\ldots,u_\delta$ are the predecessors of $u_i$.

 The hardness of the function is proved via \emph{pebbling} complexity of the graph $G$. The pebbling complexity is defined as a game of several rounds. Initially, all vertices of the graph are empty. In every round, one puts pebbles on vertices if all their predecessors already have pebbles on them from the previous round. One can remove a pebble from a vertex at any time. The game concludes when the player puts a pebble on the sink. In the sequential model of computation, only one vertex can be pebbled at each round. The cumulative pebbling complexity of a game is the sum of the number of pebbled vertices over each round. The cumulative pebbling complexity of the graph is the minimum cumulative pebbling complexity of any pebbling strategy to complete the game.

 It is natural to relate the cumulative pebbling complexity of $G$ with the upper bound of cumulative memory complexity of the memory-hard function defined over $G$. To evaluate the function, one can follow the pebbling game and compute labels of the vertices being pebbled in that round. Putting a pebble on $v_i$ with predecessors  $u_1,\ldots,u_\delta$ implies that all of $u_1,\ldots,u_\delta$ are pebbled and hence their labels are computed. Thus one can compute $\ell_i=H(i,\ell_{u_1},\ell_{u_2},\ldots,\ell_{u_\delta})$.   In the breakthrough result of \cite{STOC:AlwSer15}, Alwen and Serbinenko showed that a lower bound of cumulative pebbling complexity of $G$ also implies a lower bound on the cumulative memory complexity of the memory-hard function defined over $G$ when $H$ is a random oracle that can be accessed in a parallel fashion (more than one query at a time). All the major follow up results \cite{crypto/AlwenB16,eurosp/AlwenB17,eurocrypt/AlwenBP17} then focused on constructing graphs with high pebbling complexity. For bandwidth-hardness in \cite{TCC:RenDev17}, the authors defined the energy cost that relates to a red-blue pebbling game. 
In a red-blue pebbling game, there are two types of pebbles. A red pebble corresponds to the data in the cache, and a blue pebble models data in memory. Data in memory must be brought into the cache before using it for computation. Accordingly, a blue pebble must be replaced by a red pebble before the successors can be pebbled.

\subsection{Our Contributions}
\label{sec:our-contributions}

We formalize the notion of bandwidth-hardness in the parallel random permutation model and provide a generic construction of a data-independent random permutation based bandwidth-hard function. The main contribution is a reduction to the energy complexity of the constructed function from the red-blue pebbling complexity of the underlying depth-robust graph. The results extend the theory developed by Chen and Tessaro~\cite{C:CheTes19} who established an analogous relation for the \emph{cumulative memory complexity}~\cite{STOC:AlwSer15} of a random permutation based memory-hard function and the \emph{black} pebbling complexity of the underlying graph.


We start by considering graphs with constant in-degree $\delta$. We consider a Davis-Meyer inspired labelling function, which makes one call to an underlying random permutation. We demonstrate that the given a graph along with the labelling function will establish to a bandwidth-hard function if the red-blue pebbling complexity of the given graph is high. 

The labelling function for a constant-indegree graph makes a single query to the underlying random permutation. The labelling function $\sf{lab}$ takes input $(x_1,x_2,\ldots,x_\delta)$ where $x_1,\ldots,x_\delta$  are the labels of the $\delta$-many predecessors and outputs $\pi(\xor_{i=1}^\delta x_i)\xor \xor_{i=1}^\delta x_i $. If a node has a single predecessor, then the labelling function outputs $\pi(x)\xor x$ where $x$ is the label of the predecessor. The main technical contribution of our work is to devise a compression argument which works for any constant in-degree (\cite{C:CheTes19} crucially requires indegree to be $2$),  even when the labelling function does not take the identity of the node. The adversary can make inverse queries to the underlying primitive (in our case, random permutation).

Finally, we extend the output length of the bandwidth-hard function by extending the technique of \cite{C:CheTes19} and show a direct way to convert any depth-robust graph $G$ into a bandwidth-hard function. 

\subsubsection{Impact of our result}
\label{sec:impact-our-result}
The main analysis aims to prove that the notion of red-blue pebbling games can be applied to capture the bandwidth-hardness in the parallel random permutation model as well. While parallel random oracle models have been popular in this setting, we stress that in all practical protocols, the random oracles are instantiated using a standard hash function where the underlying building blocks are public. Indeed the Keccak-$f$, the SHA-3 standard, uses a large permutation as the building block on which the sponge mode is applied to instantiate the hash function. Moreover, as the underlying game is inherently multi-stage, we can not leverage the indifferentiablity results of the standard hash functions. Thus our work fills the crucial gap of building the theory of energy-hardness on a more realistic ideal assumption.  

We note that the construction of a random permutation based bandwidth-hard function could be dug out from two previous works. One can use the technique of \cite{C:CheTes19} to first prove memory-hardness and then apply a result from \cite{CCS:BloRenZho18}. The later work showed a lower bound of energy cost by the square root of the memory-cost. The composition results in a lower bound on $\Omega(|V|\sqrt{n})$ where $|V|$ is the number of vertices in the underlying graph, and $n$ is the output length of the random permutation. In contrast, we show a hardness lower bound of $\Omega(|V|n)$ which is a \emph{quadratic improvement} on the block size of the underlying permutation. As we show later, the improvement comes from a revised predictor algorithm where we crucially leverage the structure of the labelling function and thus avoid the need to use a generic lower bound argument. 

Finally, our random oracle instantiation is the same as in Argon2i. Thus our work is the first step in establishing proof of security for Argon2i in the random permutation model.

\subsubsection{Overview of our Technique}
\label{sec:overv-our-techn}
We construct bandwidth-hard functions from hard-to-pebble graphs. Our main theorem presents a lower bound of the energy complexity of the constructed function in terms of the red-blue pebbling complexity of the underlying graph. As explored in \cite{TCC:RenDev17,CCS:BloRenZho18}, the energy cost of function evaluation is derived based on the amount of data transfer between the cache and the main memory. Our objective is to transform any evaluation strategy of the function into a red-blue pebbling of the graph such that the red-blue pebbling complexity remains proportional to the energy cost of the evaluation.

In \cite{C:CheTes19}, the authors already showed how to transform any evaluation strategy into a \emph{black} pebbling. The idea is to examine the random permutation queries made by the evaluation in each round. We say the label for vertex $v$ is an output of round $i$, if the  random permutation query $\pi(\xor_{i=1}^\delta x_i)$ is made at the end of round $i$, where the labels $x_1,\ldots,x_{\delta}$ are of the predecessors of $v$ and were generated as output in the previous rounds. A pebble placed on the vertex $v$ in round $i$ if label of vertex $v$ is an output of round $i$ and in some future round the label of some successor of $v$ is evaluated  and label of $v$ is not recomputed in the intermediate rounds. Thus the label of $v$ must be kept in memory.

 In our case, however, we aim to derive the relations in terms of the amount of \emph{data-transfer} between cache and the memory. While the new red pebbles correspond to the random-permutation query, we need to relate the red to blue or blue to red moves with the data-transfer. The challenge is to extract such relation when the evaluator may use arbitrary data encoding. Moreover, the evaluation may indeed query the inverse functions. Thus, we can not make any direct connections between the data-transfer and the red-blue pebbling by following \cite{C:CheTes19} alone.

The idea is to establish the relationship in a somewhat amortized fashion. Following \cite{CCS:BloRenZho18}, we consider the setup where the evaluation has only $mn$ bits of cache (containing $m$ many $n$-bit words), and the pebbling has $20\delta m$ many red pebbles available. We start with a black pebbling $P=(P_1,\ldots,P_t)$. Instead of looking at each round, we partition the rounds into intervals $\{(t_0=0,t_1], (t_1,t_2],\ldots (t_{k-1},t_k=t]\}$ such that the evaluation transfers at least $mn$ bits between the cache and the memory during each interval.  We design the our red-blue pebbling in such a way so that we make at most $10\delta m$ red moves during that interval. We divide the red pebbles into two sets of size $10\delta m$ denoted as $R_i^{old}$ and $R_i^{move}$ respectively.  The set $R_i^{old}$ will be kept as red, whereas the set $R_i^{move}$ is transferred from blue pebbles. Thus the evaluation incurs energy cost to transfer $mn$ bits of data, and the red-blue pebbling has the cost of at most $20\delta m$  moves between red and blue. The final challenge is to ensure that such partitioning is always possible: ``at least $mn$ bits transferred at the cost of at most $20\delta m$ moves''. As in ~\cite{STOC:AlwSer15,C:CheTes19,CCS:BloRenZho18}, we can compress a random permutation with the help of any evaluator who bypasses such partitioning.

\section{Notations and Preliminaries}
\label{sec:notat-prel}
If $S$ is a set, $|S|$ denotes the size of $S$. $x\sample S$ denotes
the process of choosing $x$ uniformly at random from $S$. For strings $x$ and $y$, $|x|$ denotes the length of the string $x$. $x||y$ denotes the concatenation of $x$ and $y$.  $[n]$
denotes the set $\{1,\cdots,n\}$. We use boldface letters to denote vectors.

We need the following Lemma from \cite{TCC:DziKazWic11}.
\begin{lemma}
  \label{lemma:old}
  Let $B$ be a sequence of bits. Let $\alA$ be a randomized algorithm that, on input $h\in H$ for some set $H$, adaptively queries specific bits of $B$ and outputs guesses for $k$ many bits that were not queried. The probability (over the random coins of $\alA$) that there exists an $h\in H$ where $\alA(h)$ guesses all bits correctly is at most $\frac{|H|}{2^k}$.  
\end{lemma}
\subsection{Parallel Random Permutation Model}
\label{sec:parall-ideal-prim}
In this paper, we consider the ideal random permutation model for our analysis. Let $\Pi$ be the set of all permutations over $\booln$. The parallel ideal primitive model \cite{C:CheTes19} is the generalized model of the parallel random oracle model of Alwen and Serbinenko \cite{STOC:AlwSer15}. For any oracle algorithm \algo{A}, input $x$ and internal randomness $r$, the execution in the parallel ideal primitive model works in the following way. A function \algo{ip} is chosen uniformly at random from the family $\IP$. $\algo{A}$ can query $\algo{ip}$.  In our paper, $\IP$ is a family of random permutations. Thus, a permutation $\pi$ is chosen uniformly at random from $\Pi$. the forward queries are denoted by $(\smallip,+,x)$ and the response is $\pi(x)$. The inverse queries are denoted by $(\smallip,-,y)$ the response is $\pi^{-1}(y)$.

\subsection{Complexity Models}
\label{sec:complexity-models}
Throughout the paper, the notion of state is tuple $(\sigma,\zeta)$ where $\sigma$ is used to denote the content of the cache and $\zeta$ is used to denote the content of the memory. Let $\sigma_0=x, \zeta_0=\emptyset)$ be the initial state where the cached content $x$ is the input given to \algo{A} who has a cache size of $mn$ bits, and the main memory content is empty. For each round $\alA(x;r)$ takes input the cache state $\sigma_{i-1}$, performs unbounded computation and transfer data between memory and cache, and generates output cache state $\sigma'_i=(\delta_i,\mathbf{q}_i, \mathbf{out}_i)$  where $\delta_i$ is a binary string, $\mathbf{q}_i$ is a vector of queries to the random permutation $\pi$, and each element of $\mathbf{out}_i$ is of the form $(v,\ell_v)$ where $v$ and $\ell_v$ are $n$-bit labels.  After the $i^{th}$ round, $\zeta'_i$ denotes the content of the memory. We define $\sigma_i=(\delta_i,\mbox{\bf ans}(\mathbf{q}_i))$ to be the input cache state for round $i+1$, where $\mbox{\bf ans}(\mathbf{q}_i)$ is the vector of responses to queries $\mathbf{q}_i$. As the cache size is $mw$ bits, $|\mathbf{q}_i|\leq m$ as otherwise, $\alA$ will not be able to store all the responses in the cache. We say the $\alA$ terminates after round $t$ if $\mathbf{q}_t=\emptyset$.
$\alA(x;r)$ needs to transfer data between cache and memory. We assume the attacker follows some \emph{arbitrary} functions $f_1,f_2,f_3,f_4$ for communication between cache and memory during each round. The only requirement is that the functions need to be independent of the random permutation. We let $R_i=\{\rho_i^1,\rho_i^2,\cdots,\rho_i^{l_i}\}$ denote the sequence of messages sent from cache to memory during round $i$. Similarly, we let $S_i=\{s_i^1,s_i^2,\cdots,s_i^{l_i}\}$ denote the responses sent from memory to cache. 

Trace of $\alA(x;r)$ is the vector of all the input and output states of the execution. $trace(\alA(x;r))=(\sigma_0,\sigma'_1,\zeta'_1,R_1,S_1,\sigma_1,\cdots)$. Throughout the paper, we assume an upper bound, denoted by $q$, on the total number of queries made to $\pi$ by $\alA(x;r)$.

\subsubsection{Memory Complexity}
\label{sec:memory-complexity}
Given $trace(\alA^\pi(x;r))$ on input $x$, randomness $r$ and a random permutation $\pi$, the time complexity $time(\alA^\pi(x;r))$ is defined as the number of rounds before $\alA$ terminates. We define the space complexity $space^\pi(\alA(x;r))$ by the size of the maximal input state. The time complexity (resp. space complexity) of $\alA$ is the maximal time complexity (resp. maximal space complexity) overall $x,r$ and $\pi$. The cumulative memory complexity (CMC) \cite{STOC:AlwSer15} is defined as
\begin{definition}
  \label{def:cmc}
  Given  $trace(\alA^\pi(x;r))$, we define the cumulative memory complexity as
  \begin{align*}
    CMC(\alA^\pi(x;r)) =\sum_{i=0}^{time(\alA^\pi(x;r))} |\sigma_i| +|\zeta_i|
  \end{align*}
  where $\sigma_i$ denotes the input cache state of round $i$ and $\zeta_i$ denotes memory content at the start of round $i$.\\
  For a real $\epsilon\in (0,1)$, and a family of function $\mathcal{F}=\{f^\pi\colon\calx\to\caly\}_{\pi\in\Pi}$, we define the $\epsilon$-cumulative memory complexity of $\calf$ is defined to be
  \begin{align*}
    CMC_{\epsilon}(\calf)\eqdef \mbox{min}_{x\in\calx,\alA\in \cala_{x,\epsilon}} \Ex[CMC(\alA^\pi(x;r))]
  \end{align*}
  where $ \cala_{x,\epsilon}$ is the set of all parallel algorithms that with probability at least $\epsilon$ on input $x$ and oracle access to $\pi$, output $f^\pi(x)$.    Here, the probability and expectation are calculated over the randomness of the random permutation $\pi$ and the internal randomness of $\alA$.
 \end{definition}

\subsubsection{Energy Complexity}
\label{sec:energy-complexity}
\begin{definition}
  \label{def:cmc}
  Given  $trace(\alA^\pi(x;r))$, we define the energy complexity as
  \begin{align*}
    cost(\alA^\pi(x;r)) =\sum_{i=0}^{time(\alA^\pi(x;r))} \left(c_r|\mathbf{q}_i|+ \sum_{j=1}^{l_i} c_b (|\rho_i^j|+|s_i^j|)\right)
  \end{align*}
  where $c_r$ denotes the cost of random permutation query and $c_b$ denotes the cost for data transfer between the cache and the main memory. \\
  For fixed $c_b,c_r$, a real $\epsilon\in (0,1)$, and a family of function $\mathcal{F}=\{f^\pi\colon\calx\to\caly\}_{\pi\in\Pi}$, we define the $\epsilon$-energy complexity of $\calf$ is defined to be
  \begin{align*}
    ecost_{\epsilon}(\calf,mn)\eqdef \mbox{min}_{x\in\calx,\alA\in \cala_{x,\epsilon}} \Ex[cost(\alA^\pi(x;r))]
  \end{align*}
  where $ \cala_{x,\epsilon}$ is the set of all parallel algorithms with at most $mn$ bits of cache that with probability at least $\epsilon$ on input $x$ and oracle access to $\pi$, output $f^\pi(x)$.
\end{definition}

\subsection{Memory and Bandwidth Hardness}
\label{sec:memory-bandw-hardn}
In this section, we recall the definition of Memory and Bandwidth hardness.
\subsubsection{Memory Hard Functions}
\label{sec:memory-hard-funct}
We now recall the definition of memory-hardness in the random permutation model. The definition follows the spirit of the definition provided in \cite{C:CheTes19}. Intuitively a function is memory-hard if there exists a somewhat efficient sequential algorithm that computes the function, whereas any parallel algorithm that computes the function correctly must pay a high CMC cost.

\begin{definition}
    \label{def:mhf-orig}
    Consider a family of function $\mathcal{F}=\{f^\pi\colon\calx\to\caly\}_{\pi\in\Pi}$. The family $\mathcal{F}$ is $(\epsilon, \delta, q)$ memory hard if for all evaluation point $x$ and an oracle-aided algorithm $\alA^\pi$ that makes at most $q$ many $\pi$ queries such that:
    \begin{itemize}
        \item $\Pr[\alA^\pi(x) = f^\pi(x)] > \epsilon$
        \item  $CMC_{\epsilon}(\calf)\geq \delta$
    \end{itemize}
    Here, the probability and expectation are calculated over the randomness of the random permutation $\pi$ and the internal randomness of $\alA$.
\end{definition}
\subsubsection{Bandwidth-Hard Functions}
\label{sec:bandw-hard-funct-1}
We now extend the above definition of memory-hardness to the bandwidth-hardness in the random permutation model.
\begin{definition}
    \label{def:mhf-orig}
     Consider a family of function $\mathcal{F}=\{f^\pi\colon\calx\to\caly\}_{\pi\in\Pi}$. The family $\mathcal{F}$ is $(\epsilon, \delta, q)$ bandwidth-hard if for all evaluation point $x$ and all oracle aided algorithm $\alA^\pi$ that makes at most $q$ many $\pi$ queries and uses at $mn$ bits of cache it holds that
    \begin{itemize}
        \item $\Pr[\alA^\pi(x) = f^\pi(x)] > \epsilon$
        \item  $ecost_{\epsilon}(\calf,mn)\geq \delta$
    \end{itemize}
    Here, the probability and expectation are calculated over the randomness of the random permutation $\pi$ and the internal randomness of $\alA$.
\end{definition}  
\subsection{Graphs and Pebbling Models}
\label{sec:graphs-pebbl-models}
We use $G=(V,E)$ to denote a directed acyclic graph (DAG) with $|V|=2^n$ nodes, Let ${\sf src}(V)\subseteq V$ be the set of source nodes, and ${\sf snk}(V)\subseteq V$ be the set of sink nodes. For a node $v$, ${\sf pred}(v)\eqdef\{u\in V|(u,v)\in E\}$ are the predecessors of $v$, and ${\sf succ}(v)\eqdef\{w\in V|(v,w)\in E\}$ are the successors of $v$. We use $ind(v)=|{\sf pred}(v)|$ as the indegree of $v$. For a directed acyclic path $P$, the length of $P$ is the number of nodes it traverses. ${\sf depth}(G)$ is the length of the longest acyclic path in $G$. For a source node $v$, $ind(v)=0$. For $S\subseteq V$, ${\sf pred}(S)\eqdef \cup_{v\in S}{\sf pred}(v)$.

\begin{definition}
  A DAG $G=(V,E)$ is $(e,d)$- depth-robust if and only if ${\sf depth}(G\setminus S)\geq d$ for any $S\subseteq V$ where $|S|\leq e$. Moreover, $G$ is said to be $(e,d)$-source-to-sink depth-robust if and only if for any $S\subseteq V$ where $|S|\leq e$, $G\setminus S$ has a path of length at least $d$ from a source node to a sink node in $G$.
\end{definition}

\subsubsection{Pebbling}
\label{sec:pebbling}
The pebbling game is played on a DAG $G$ in the mode of parallel pebbling. A pebbling of a DAG $G$ is a sequence of pebbling configurations $P=(P_0,\cdots,P_t)$ where $P_0=\emptyset$ and $P_i\subseteq V$ for all $i\in [t]$. A pebbling is legal if for any $i\in [t]$, for any $v\in P_i\setminus P_{i-1}$, ${\sf pred}(v)\in P_{i-1}$.  Finally, a pebbling is successful, if for every sink node $v_s$, there exists $i\in [t]$, $v_s\in P_i$.

\subsubsection{Red-Blue Pebbling}
\label{sec:red-blue-pebbling}

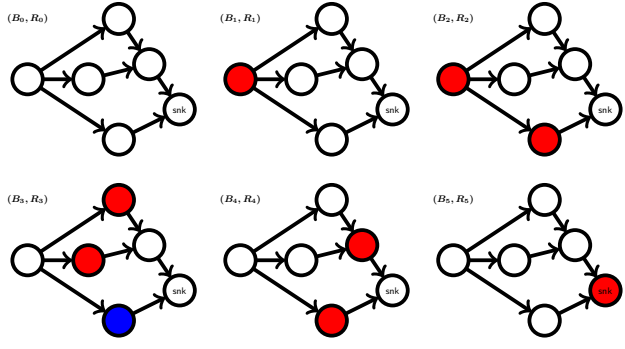
\begin{figure}
  \centering
  \begin{tikzpicture}[ thick,scale=0.4, every node/.style={transform shape},blank/.style = {draw, circle,line width=0.5mm,minimum size=1cm},red/.style = {blank, fill=red},blue/.style = {blank,,fill=blue},to/.style={->,line width=0.5mm}]

\coordinate (shift1) at (7,0);
\coordinate (shift2) at (14,0);
\coordinate (shift3) at (0,-6);
\coordinate (shift4) at (7,-6);
\coordinate (shift5) at (14,-6);

\begin{scope}
	\node [font=\boldmath] at (0,2) { $(B_0,R_0)$};
	\node[blank] (1) at (0,0) {}; 
	\node[blank] (2) at (2,0) {}; 
	\node[blank] (3) at (3,2) {}; 
	\node[blank] (4) at (3,-2) {}; 
	\node[blank] (5) at (4,0.5) {}; 
	\node[blank] (6) at (5,-1) {${\sf snk}$}; 
	\draw[to] (1) -- (2);
	\draw[to] (1) -- (3);
	\draw[to] (1) -- (4);
	\draw[to] (2) -- (5);
	\draw[to] (3) -- (5);
	\draw[to] (5) -- (6);
	\draw[to] (4) -- (6);
\end{scope}

\begin{scope}[shift =(shift1)]
	\node [font=\boldmath] at (0,2) { $(B_1,R_1)$};
	\node[red] (1) at (0,0) {}; 
	\node[blank] (2) at (2,0) {}; 
	\node[blank] (3) at (3,2) {}; 
	\node[blank] (4) at (3,-2) {}; 
	\node[blank] (5) at (4,0.5) {}; 
	\node[blank] (6) at (5,-1) {${\sf snk}$}; 
	\draw[to] (1) -- (2);
	\draw[to] (1) -- (3);
	\draw[to] (1) -- (4);
	\draw[to] (2) -- (5);
	\draw[to] (3) -- (5);
	\draw[to] (5) -- (6);
	\draw[to] (4) -- (6);
\end{scope}

\begin{scope}[shift =(shift2)]
	\node [font=\boldmath] at (0,2) { $(B_2,R_2)$};
	\node[red] (1) at (0,0) {}; 
	\node[blank] (2) at (2,0) {}; 
	\node[blank] (3) at (3,2) {}; 
	\node[red] (4) at (3,-2) {}; 
	\node[blank] (5) at (4,0.5) {}; 
	\node[blank] (6) at (5,-1) {${\sf snk}$}; 
	\draw[to] (1) -- (2);
	\draw[to] (1) -- (3);
	\draw[to] (1) -- (4);
	\draw[to] (2) -- (5);
	\draw[to] (3) -- (5);
	\draw[to] (5) -- (6);
	\draw[to] (4) -- (6);
\end{scope}

\begin{scope}[shift =(shift3)]
	\node [font=\boldmath] at (0,2) { $(B_3,R_3)$};
	\node[blank] (1) at (0,0) {}; 
	\node[red] (2) at (2,0) {}; 
	\node[red] (3) at (3,2) {}; 
	\node[blue] (4) at (3,-2) {}; 
	\node[blank] (5) at (4,0.5) {}; 
	\node[blank] (6) at (5,-1) {${\sf snk}$}; 
	\draw[to] (1) -- (2);
	\draw[to] (1) -- (3);
	\draw[to] (1) -- (4);
	\draw[to] (2) -- (5);
	\draw[to] (3) -- (5);
	\draw[to] (5) -- (6);
	\draw[to] (4) -- (6);
\end{scope}

\begin{scope}[shift =(shift4)]
	\node [font=\boldmath] at (0,2) { $(B_4,R_4)$};
	\node[blank] (1) at (0,0) {}; 
	\node[blank] (2) at (2,0) {}; 
	\node[blank] (3) at (3,2) {}; 
	\node[red] (4) at (3,-2) {}; 
	\node[red] (5) at (4,0.5) {}; 
	\node[blank] (6) at (5,-1) {${\sf snk}$}; 
	\draw[to] (1) -- (2);
	\draw[to] (1) -- (3);
	\draw[to] (1) -- (4);
	\draw[to] (2) -- (5);
	\draw[to] (3) -- (5);
	\draw[to] (5) -- (6);
	\draw[to] (4) -- (6);
\end{scope}

\begin{scope}[shift =(shift5)]
	\node [font=\boldmath] at (0,2) { $(B_5,R_5)$};
	\node[blank] (1) at (0,0) {}; 
	\node[blank] (2) at (2,0) {}; 
	\node[blank] (3) at (3,2) {}; 
	\node[blank] (4) at (3,-2) {}; 
	\node[blank] (5) at (4,0.5) {}; 
	\node[red] (6) at (5,-1) {${\sf snk}$}; 
	\draw[to] (1) -- (2);
	\draw[to] (1) -- (3);
	\draw[to] (1) -- (4);
	\draw[to] (2) -- (5);
	\draw[to] (3) -- (5);
	\draw[to] (5) -- (6);
	\draw[to] (4) -- (6);
\end{scope}
\end{tikzpicture} 
  \caption{Example of a red-blue pebbling that uses at most $2$ simultaneous red pebbles}
  \label{fig:red-blue}
\end{figure}
In case of red-blue pebbling, we consider a sequence of pebbling configurations $\mathcal{RB}=((B_0,R_0),(B_1,R_1),\cdots,(B_t,R_t))$. As before, the game is played in rounds. $B_i\subseteq V$ denote the set on which blue pebbles are placed. $R_i\subseteq V$ denote the set on which red pebbles are placed. Initially, no node is pebbled; $R_0\cup B_0=\emptyset$. The final goal is to place a red pebble on every sink node; ${\sf snk}(V)\subseteq \cup_i R_i$. The legal pebbling rule with $m$ red pebbles is in every round $i>0$; it should hold that
\begin{itemize}
\item ${\sf pred}(R_i\setminus (R_{i-1}\cup B_{i-1}))\subseteq R_{i-1}$.
\item $B_i\setminus B_{i-1} \subseteq R_{i-1}$
 \item $|R_i|\leq m$. 
\end{itemize}

We say a pebbling sequence $\mathcal{RB}=((B_0,R_0),(B_1,R_1),\cdots,(B_t,R_t))$ is sequential if in addition to the above conditions, it holds that $|R_i\setminus R_{i-1}| \leq 1$ for $0<i\leq t$. The number of Blue moves and Red moves at round $i$ are defined by the following  $BM_i$ and $RM_i$ respectively
\begin{align*}
  BM_i\eqdef |\{v \in R_i\setminus R_{i-1}: {\sf pred}(v) \not\subset R_{i-1}\}|+ |B_i\setminus B_{i-1}|\\
  RM_i\eqdef | R_i\setminus R_{i-1}|-|\{v \in R_i\setminus R_{i-1}: {\sf pred}(v) \not\subset R_{i-1}\}|
\end{align*}

\begin{definition}
  Let $\mathcal{RB}=((B_0,R_0),(B_1,R_1),\cdots,(B_t,R_t))$ be a red-blue pebbling. 
  For fixed $c_b$ and $c_r$, we define the energy cost of $\mathcal{RB}$ is defined as
  \begin{align*}
    cost(\mathcal{RB})\eqdef \sum_{i=1}^t c_b\cdot BM_i+ c_r\cdot RM_i
  \end{align*}
  
  Given a DAG $G$ and the number of red pebble $m$, we define the red blue pebbling cost of a graph as
  \begin{align*}
    rbcost(G,m)\eqdef \mbox{min}_{\mathcal{RB}\in \mathsf{RB}(G,m)}  cost(\mathcal{RB})
  \end{align*}
  where $ \mathsf{RB}(G,m)$ is the set of all red-blue pebbling of $G$ with $m$ red pebbles.
\end{definition}

Figure \ref{fig:red-blue} shows an example of red blue pebbling for a graph with $5$ nodes that uses a maximum of $2$ simultaneous red pebbles.

\subsubsection{Graph based Bandwidth-Hard Functions in the random permutation model}
\label{sec:graph-based-bandw}
For a DAG $G=(V,E)$ with $|V|=2^w$ nodes and $ind(v)\leq \gamma$ and with a set of sink nodes ${\sf snk}(V)$ and a random permutation $\pi:\booln\to\booln$, we define the labeling function of the graph $G$ with respect to an input $x=(x_1,x_2,\cdots,x_{n_s})\in \bool^{n_sn}$ as $\sf{lab}_{G,\pi,x}\colon V\to \bool^{n}$ which is recursively defined as
\begin{align*}
  \sf{lab}_{G,\pi,\tau,x}(v)=\twopartdef{\pi(x_i)\xor x_i}{v\mbox{ is the } i^{th}\mbox{ source }}{\tau(\sf{lab}_{G,\pi,\tau,x}(v_1),\sf{lab}_{G,\pi,\tau, x}(v_2),\cdots \sf{lab}_{G,\pi,\tau,x}(v_\gamma))} {ind(v)>0}
\end{align*}
where $\tau^\pi\colon \bool^{\gamma n}\to\booln$ is the labeling function for non-source nodes, ${\sf pred}(v)=\{v_1,\cdots,v_\gamma\}$. We define $f_{G,\pi,\tau}(x)=\{ \sf{lab}_{G,\pi,\tau,x}(v_s)\}_{v_s\in {\sf snk}(V)}$ as the graph function.

\section{Bandwidth-Hard Functions in the random permutation model.}
\label{sec:bandw-hard-funct}

In this section, we construct a family of graph-based bandwidth-hard functions from $n$-bit permutation. Specifically, we show that the labelling functions created in the previous section are bandwidth-hard. Given $\pi\in\RP$ and a graph $G$, we defined the non source nodes labeling function $\tau(x_1,x_2)\eqdef\pi(x_1\xor x_2)\xor x_1\xor x_2$. If a non-source node has a single predecessor, then $\tau(x)=\pi(x)\xor x$.

The construction can be generalized for DAGs with maximum indegree $\delta>2$. Our proofs are done for a constant $\delta$. We note however, the black pebbling proved in \cite{C:CheTes19} considered $\delta=2$. We say an input vector $x=(x_1,\cdots,x_{n_s})$ is non-colliding if for all distinct $i,j$ we have $x_i\neq x_j$.

\begin{theorem}
  \label{thm:smallblock}
  Consider a random permutation $\pi$ over $\booln$.  Fix a $\delta$-indegree predecessor distinct DAG $G=(V,E)$  Assume an adversary can make no more than $q$ oracle calls  output calls such that $q=2^{n/10\delta}$. Consider the graph function $f_{G,\pi,\tau}$ with $\tau$ being the non-source nodes' labelling function. Moreover, suppose the input to $f_{G,\pi,\tau}$ is non-colliding. Then there exists $\epsilon\in (0,1]$ such that  if $|V|\leq 2^{n/4\delta}$, it holds that
  \begin{align*}
          ecost(f_{G,\pi,\tau},mn) \geq \frac{\epsilon}{40\delta}rbcost(G,20\delta m)-\frac{\epsilon mc_b}{2}. 
      \end{align*}
\end{theorem}



\subsection{Proof of Theorem~\ref{thm:smallblock}}
\label{sec:proof-theorem-smallblock}
\noindent\underline{Label Notations}.
Fix an input vector $x$ and the underlying DAG $G=(V,E)$. For any node $v\in V$, we denote by ${\sf lab}(v)$ the graph label of $v$. For every node $v$, by the term pre-label of $v$ ( denoted by ${\sf prelab}(v)$) we define the input of the $\pi$ query to compute ${\sf lab}(v)$. If $v$ is a source, then by definition ${\sf prelab}(v)=x_v$. For a non-source node $v$, ${\sf prelab}(v)= \xor_{i=1}^{d(v)} {\sf lab}(v_i)$ where $v_i\in \mathsf{pred}(v)$ are the predecessors of $v$. For every node $v\in V$, we define ${\sf postlab}(v)=\pi({\sf prelab}(v))$. By construction, ${\sf lab}(v)={\sf postlab}(v)\xor {\sf prelab}(v)$. 

\noindent{\textbf{Labels of the nodes are distinct.}} The first property we need from the labelling is that all the node labels are distinct during the evaluation of the labelling function for a given input. Thus looking at the label, we can identify the corresponding node.
 \begin{myclaim}
    \label{claim:nocoll}
    Suppose $G=(V,E)$ be a DAG with $n$ vertices such that ${\sf pred}(u)\neq {\sf pred}(v)$ for all distinct $u,v\in V$. Let ${\tt Coll}$ denote the event that there exists two distinct nodes $u,v$ with ${\sf lab}(u)={\sf lab}(v)$ or ${\sf prelab}(u)={\sf prelab}(v)$ during one evaluation. It holds that
    \begin{align*}
      \Pr_\pi\left[{\tt Coll}\right] \leq \frac{2|V|^2}{2^n}
    \end{align*}
  \end{myclaim}



\subsubsection{Extending Black Pebbling to Red-Blue Pebbling}
\label{sec:extension-pebbling}
 Consider the black pebbling guaranteed by \cite{C:CheTes19}.
We start with the notion of extension red-blue pebbling in the ideal primitive model. Given a DAG $G$ and a legal black pebbling $P=(P_0,\cdots,P_t)$ with $|P_{i+1}\setminus P_i| \leq m$, we say that a (legal) red-blue pebbling is a $k$-extension of $P$ if $\forall i\in [t]$, we can find a small $E_i\subseteq V$ such that $|E_i|\leq k$, $P_i\subseteq B_i\cup R_i$, and in particular $P_i\cup E_i= B_i\cup R_i$.

\subsubsection*{Correct and Critical Calls}
\label{sec:corr-crit-calls}
The idea of the correct call is to point out the query corresponding to the evaluation of a node. We note however in \cite{C:CheTes19}, a correct call for a node $v$ in round $i$ does not automatically pebble node $v$ in round $i$. A query to a node leads to the node getting pebbled only if it is a sink node or its label is used in the future. The notion of critical query captures this idea.

In the random permutation model, the algorithm $\alA$ can make two types of queries to the permutation oracle. For a forward query $\pi(x)$, we say it is a \emph{correct} call for a node $v\in V$ if and only if it holds that $x={\sf prelab}(v)$. Similarly, the query $\pi^{-1}(x)$ is correct call for vertex $v$ if and only if $x={\sf postlab}(v)$.\\

Now we define critical calls. A forward query $\pi(x)$ is critical for a node $u\in V$ if and only if  $\exists v\in {\sf succ}(u)$ such that ${\sf prelab}(v)=x$ and $\alA$ has made no correct call for $u$ after round $i$. Similarly an inverse query $\pi^{-1}(y)$ is critical for a node $u\in V$ in round $i'>i$ if and only if  $\exists v\in {\sf succ}(u)$ such that ${\sf prelab}(v)=x$ and $\alA$ has made no correct call for $u$ after round $i$.
Additionally, if $v$ is a sink, then the first correct call for node $v$ is critical for $v$.

In the following paragraph, we extend it to critical pebbling. 
A pebbling $P_i$ at node $(v)$ is critical for $u$ with interval $[t_1,t_2]$ if
\begin{itemize}
\item  $v\in {\sf succ}(u)$.  
\item $\pi({\sf prelab}(v))$ or $\pi^{-1}({\sf postlab}(v))$ is queried at round $i$.
 \item  $\forall j\in [t_1,i-1]$ no correct call for $u$ is made.
 \end{itemize}
 We put the node $u$ in the set $\mbox{\bf Critical}(t_1,t_2)$. Formally, given a black pebbling $P$ and an interval $(t_1,t_2)$, we define
\begin{align*}
  \mbox{\bf Critical}(t_1,t_2)\eqdef \cup_{i=t_1}^{t_2} \left(\mbox{parents}(P_i\setminus P_{i-1})\setminus \left(\cup_{j=t_1}^{i-1} P_j\setminus P_{j-1}\right) \right)
\end{align*}

Suppose $u$ is a node in the set $\mbox{\bf Critical}(t_1,t_2)$ but $u\notin \mbox{\bf Critical}(t'_1,t_2)$ with $t'_1< t_1$. This implies that there is a correct call for $u$ in some round $j$ with $t'_1\leq j <t_1$.  

\subsubsection*{Partitioning Intervals}
\label{sec:part-interv}

Now we partition the $t$ pebbling rounds into intervals $(t_0=0,t_1],(t_1,t_2],\cdots$
recursively as follows. Let $t_1$ be the minimum pebbling round such that there exists $j<t_1$ such that $|\mbox{\bf Critical}(j,t_1)| > (10\delta-1)m$ . If no such $t_1$ exists, then we conclude that $|\mbox{\bf Critical}(t_0,t)| \leq (10\delta-1)m$. In that case, we propose a red-blue extension-pebbling that requires $0$ blue move and at most $\sum_i |P_i\setminus P_{i-1}|$ red moves.

Now, once we have define $t_1,t_2,\cdots,t_{i-1}$, we define $t_{i}$ to be the minimum index such that there exists $t_{i-1}<j<t_i$ such that $|\mbox{\bf Critical}(j,t_i)| \geq (10\delta-1)m$. If no such $t_i$ exists then we set $t_i=t$ and conclude the partition. Note that the size of all the Critical sets for the last partition is less than $(10\delta-1)m$.\\

\subsubsection*{Red-Blue Pebbling}
\label{sec:red-blue-pebbling}
   
Now we shall construct an extension pebbling that makes at most $10\delta m$ red moves and $10\delta m$ blue moves during each interval. Towards this we define the extension red blue pebbling $(B^*,R^*)$ by dividing the cache into two sets of size $10\delta m$ denoted as $R_i^{old}$ and $R_i^{move}$ respectively.  The set $R_i^{old}$ will be kept in the cache, whereas the set $R_i^{move}$ will be brought from memory to the cache. We will set  $R_i=R_i^{move}\cup R_i^{old}$ and $B_i=P_i$. Note that a node may contain both red and blue pebbles simultaneously as the same data may be both in the cache and the memory.

\noindent At the start of each interval $(t_i,t_{i+1}]$, we set $R_i^{old}=\emptyset$. For each $j\in (t_i,t_{i+1}] $, we define
\begin{align*}
  R_j^{old}=(R_{j-1}^{old}\cup (P_j\setminus P_{j-1}))\cap \mbox{\bf Critical}(j,t_i)
\end{align*}

Intuitively, $R_j^{old}$ stores all of the red-pebbles we compute during the interval $(t_i,j]$ that are needed in the interval $[j+1,t_{i+1})$. In other words any node whose label is computed during rounds $(t_i,j]$ that are later needed for the interval $(j+1,t_{i+1})$ will be in $R_j^{old}$, which will be kept in cache. The following claim holds directly from the definition.

\begin{claim}
  For any $j\in (t_i,t_{i+1})$ it holds that
  \begin{align*}
    \mbox{\bf Critical}(j+1,t_{i+1}) \cap (\cup_{j'=t_i}^{j} (P_{j'}\setminus P_{j'-1})) \subseteq R_j^{old}
  \end{align*}
\end{claim}

To maintain legality across all time steps, we add a few rules about blue moves.
\begin{enumerate}
\item A pebbled node $v\in R_j$ from red to blue at time $j$ if node $v$ is in $\mbox{\bf Critical}(t_i, t_{i+1})$ for some later interval $(t_i, t_{i+1})$ with $j<t_i$ and if $v\notin B_j$ is not already stored in memory.  Note, we will consider the cost $c_b$ of this blue move for the corresponding future interval $(t_i, t_{i+1})$.
\item A pebbled node $v$ is converted from blue to red if $v\in \mbox{\bf Critical}(t_i, t_{i+1})$.  In other words, we define $R_j^{move}=\mbox{\bf Critical}(t_i, t_{i+1})$.

\end{enumerate}
Now we are ready to bound $|\mbox{\bf Critical}(t_i, t_{i+1})|$.
\begin{lemma}
  \label{lemma:size}
  \begin{align*}
    \forall~j\in(t_{i},t_{i+1}); \lvert\mbox{\bf Critical}(j, t_{i+1})\rvert\leq 10\delta m
  \end{align*}
\end{lemma}
\begin{proof}
  By construction of the interval, $|\mbox{\bf Critical}(j, t_{i+1}-1)|\leq (10\delta -1)m$. As cache size is bounded by $m$, $|{\sf pred}(P_{t_{i+1}}\setminus P_{t_{i+1}-1})|\leq m$. Thus  $\lvert\mbox{\bf Critical}(j, t_{i+1})\rvert\leq 10\delta m$. \qed
\end{proof}

\begin{lemma}
  $R_i$ is a legal pebbling.
\end{lemma}
\begin{proof}
  We start from the observation that ${\sf pred}(P_{j+1}\setminus P_{j})\subseteq \mbox{\bf Critical}(j, t_{i+1})$. For any $v\in \mbox{\bf Critical}(j, t_{i+1})$, either $v\in \mbox{\bf Critical}(t_i, t_{i+1})$ (thus $v\in R_j^{move}$) or $v$ has been pebbled at some step within the interval $(t_i,j)$ (thus $v\in R_j^{old}$). As $R_j=R_j^{move}\cup R_j^{old}$, we conclude ${\sf pred}(P_{j+1}\setminus P_{j})\subseteq R_j$. Hence all the parent nodes are in the cache, and hence the pebbling is legal. \qed
\end{proof}

\begin{lemma}
  $|R_j^{old}|\leq 10\delta m$. 
\end{lemma}
\begin{proof}
 As $R_j^{old}\subseteq \mbox{\bf Critical}(j+1, t_{i+1}) $ and by Lemma \ref{lemma:size}, $|\mbox{\bf Critical}(j+1, t_{i+1})|\leq 10\delta m$, the lemma follows.\qed
\end{proof}

To bound the cost of the above extension pebbling we observe that only cache misses that we need to consider are the ones in $R_{t_i}^{move}$ whose size is $|\mbox{\bf Critical}(t_i, t_{i+1})|\leq 10 \delta m$. Considering their movement from the cache to memory and back to cache, the total cost due to the cache misses is $2\times 10\delta m\times c_b$, which is equal to $20\delta mc_b$ where $c_b$ is the cost of cache to memory data transfer. Thus the total cost of the extension pebbling for the interval $(t_i,t_{i+1}]$ is bounded above by
\begin{align*}
  20\delta mc_b+\sum_{j\in (t_i,t_{i+1}]} c_r (P_{j}\setminus P_{j-1})
\end{align*}

\begin{remark}
  We stress that we put a red pebble on a node only when there is a critical query for that node. This strategy is identical to the labelling principle of \cite{CCS:BloRenZho18}. However, we work with fewer red pebbles as similar to\cite{C:CheTes19}  we work with the pebbling extraction such that the sink nodes are put in the critical set and pebbled as soon as their labels are evaluated.
\end{remark}

\subsection{The Predictor}
\label{sec:predictor}

Next we show that for every interval any algorithm $\alA$ needs to pay $mc_b$ cost for blue moves in addition to $\sum_{j\in (t_i,t_{i+1}]} c_r (P_{j}\setminus P_{j-1})$ cost for blue moves. In other words, we show that for an interval, $\alA$ transfers at least $mn$ bits between the cache and the memory. If such an algorithm exists, then that algorithm can be converted into a predictor for the random permutation $\pi$ resulting in a compression algorithm for a random permutation.
\begin{lemma}
  \label{lemma:predictor}
  Fix $n$, $\delta$-indegree predecessor-distinct DAG $G=(V,E)$ with $|V|\leq 2^{n/8\delta} $ and $n_s$ many source nodes, non-colliding input vector $x\in\bool^{n_s n}$, algorithm $\adv$ (that makes $q \leq 2^{n/8\delta}$ many calls).  Define $Bad$ as the event where all the following conditions are satisfied
  \begin{enumerate}
  \item The pre-labels are distinct.
  \item The red-blue pebbling is legal.
    \item There exists an $i\in \NN$ such that for interval $(t_i,t_{i+1})$, the interval is not the last one, and the algorithm $\adv$ sends less than $mn$ bits between the cache and the memory.  
  \end{enumerate}
  It holds that $\Pr[Bad]\leq 2^{-2mn+1}+ 2^{-n(1-1/4\delta)}$.
\end{lemma}

\subsubsection{Proof of Lemma~\ref{lemma:predictor}}
\label{sec:proof-lemma-predictor}

Without loss of generality, we fix the random string $r$ to be the random coins of $\adv$ that maximized $\Pr[Bad]$. We shall show a predictor $\algo{P}$ (that hardwires the random string $r$ and has oracle access to $\pi$.) such that if $Bad$ happens, there will be a hint $h$ of size less than $(10\delta-1)mn$ bits such that $\algo{P}(h)$ can predict the outputs of at least $(10\delta-1)m$ many points of $\pi$ correctly. The bound contradicts the incompressibility of a random permutation, and the lemma will hold. First, we describe the hint $h$  and the predictor $\algo{P}$.

\textsc{The Hints.} Let $i\in\NN$ be the integer such that, for the interval $(t_i,t_{i+1})$, $\adv$ transfers less than $mn$ bits between the cache and the memory. This hint is going to consist of the following components.
\begin{enumerate}
\item The set $\mbox{\bf Critical}(t_i, t_{i+1})$ is given as the points the extractor needs to predict. Recall that $\lvert \mbox{\bf Critical}(t_i, t_{i+1})\rvert\leq 10\delta m$ and hence the first component size of the hint is $10 \delta m\log |V|$. We assume a topological ordering of the vertices in the critical set based on the order of the critical calls. In other words, we order the elements of $\mbox{\bf Critical}(t_i, t_{i+1})$ as $(v_1,v_2,\cdots,)$. 
\item A sequence $Q=(id_1,id_2,\cdots,id_{\lvert \mbox{\bf Critical}(t_i, t_{i+1})\rvert})$ such that for each $v_j\in\mbox{\bf Critical}(t_i, t_{i+1}) $, $id_j$ is the index of the first critical query for $v_j$. This hint component tells the extractor the queries that require the prediction of the output of the random permutation. The total size of this hint is at most $10\delta m\log q$ bits.
\item A sequence $W=(w_1,w_2,\cdots,w_{\lvert \mbox{\bf Critical}(t_i, t_{i+1})\rvert})$ of nodes where such that $id_j$ is a correct call for  $w_j$ and $w_j$ is some successor of $v_j\in \mbox{\bf Critical}(t_i, t_{i+1})$. The total size of this hint is at most $10\delta m\log |V|$.
\item A sequence $L=(id'_1,id'_2,\cdots,id'_{\lvert \mbox{\bf Critical}(t_i, t_{i+1})\rvert})$ such that $id'_j$ is the first correct call for  $v_j\in\mbox{\bf Critical}(t_i, t_{i+1}) $. This hint is required to ensure that the predictor does not query the oracle on these inputs. The size of $L$ is clearly bounded by $10\delta m\log q$. 
\item   $H=\{h_j|v_j \in \mbox{\bf Critical}(t_i, t_{i+1})\}$, where $h_j$ is the label ${\sf lab}({v_j})$ if there exists some  $k>j$ such that $id_k=id_j$, otherwise $h_j$ is empty. Size of $H$ is at most $10\delta m(1-\frac{1}{\delta})n$ bits.
  
\item The cache state at $t_i$ is given as a hint to simulate the attacker from the time slot $t_i$. In addition the messages between the cache and the memory during $(t_i,t_{i+1})$ is also given. The cache size is bounded by $mn$ bits. The total number of bits required for the messages passed between the cache and the memory is assumed to be less than $mn$ bits, as premised in the definition of the bad event.
\end{enumerate}
In total, the size of the hint is less than

\begin{align*}
  10\delta m (2\log |V|+ 2\log q)+  10\delta(1-\frac{1}{\delta})mn +2mn.
\end{align*}

Putting the conditions $|V|\leq 2^{n/8\delta}$ and $q\leq 2^{n/8\delta}$, we bound

\begin{align*}
  2\log |V|+ 2\log q \leq \frac{n}{4\delta} + \frac{n}{4\delta}= \frac{n}{2\delta} 
\end{align*}
Then the size of the hint is less than
\begin{align*}
  5mn + 10\delta mn -10mn+2mn  
\end{align*}
Thus the total hint size is less than $(10\delta-3)mn$.

\textsc{Simulating $\alA$}. Given an input, the predictor $\algo{P}$ recovers the pebbling configuration $P_i$ and parses the input into the hints described above. Then $P$ runs $\alA(\sigma_i)$ and attempts to \emph{ predict} $({\sf prelab}(v),\pi({\sf prelab}(v)))$ for \emph{every} $v\in \mbox{\bf Critical}(t_i, t_{i+1})$. When simulating  $\alA(\sigma_i)$, the algorithm $\algo{P}$ first needs to figure out if the ideal primitive call is a correct call for a vertex $v$. Towards this the predictor keeps track of the labels ${\sf prelab}(u)$, ${\sf postlab}(u)$, and ${\sf lab}(u)$ for all $u\in V$. Once a node's predecessors' labels are fixed, the predictor accordingly updates the {\sf prelab} of that node. Given an ideal-primitive call from $\adv$, $\algo{P}$ performs the following.
\begin{itemize}
\item If $\pi(x)$ is the first critical call for some node $v$ of the critical set (this can be confirmed from $Q$), then $\algo{P}$ finds the node $w\in W$ such that $x={\sf prelab}(w)$. If ${\sf lab}(w)$ is in the list $\alA$ returns ${\sf lab}(w)\xor x$ as the response. 
\item If the call is an inverse call $\pi^{-1}(y)$, then the predictor checks if there is a node $v\in V$ such that ${\sf postlab}(v)=y$ and ${\sf lab}(v)$ is in the list. If yes, $\alA$ returns ${\sf lab}(v)\xor y$. If no such $v$ exists, $\algo{P}$ queries the oracle and checks if the answer is consistent with some ${\sf prelab}(v)$ for $v\in V$. $\algo{P}$ updates ${\sf prelab}(v),{\sf postlab}(v),{\sf lab}(v)$.
\end{itemize}

\textsc{Handling Critical calls}.
Next, we focus on predicting the output of the oracle for the critical queries. The predictor simulates $\adv(\sigma_i)$. For each round $\gamma>i$, after receiving the calls from $\adv$, the predictor $\algo{P}$ responds to the critical calls in the following way

\begin{itemize}
\item $\algo{P}$ first enumerates node $v_j\in \mbox{\bf Critical}(t_i, t_{i+1})$ according to reverse topological order and checks the following: If the $id_j$-th call is in round $\gamma$ and ${\sf lab}({v_j})$ is unknown yet, the predictor uses the hint to predict the output. Note that the call is correct for node $w_j\in W$. The first step for $\algo{P}$ is to find the ${\sf prelab}(w_j)$. Suppose $w_j\in \mbox{\bf Critical}(t_i, t_{i+1})$. If the call is a forward call $\pi(x)$ then $\algo{P}$ sets $x={\sf prelab}(w_j)$. If the call is an inverse call $\pi^{-1}(y)$, then by the fact that $v_j$ is chosen in the reverse topological order, and the critical call for $w_j$ is made before any correct call for $w_j$, the label ${\sf lab}({w_j})$ is already fixed. Now, $\algo{P}$ extracts ${\sf prelab}(w_j)={\sf lab}({w_j})\xor y$ without calling the oracle. If  $ w_j\notin \mbox{\bf Critical}(t_i, t_{i+1})$, $\algo{P}$ can query $\pi^{-1}(y)$ to get ${\sf prelab}(w_j)$.
\item Once $w_j$ and ${\sf prelab}(w_j)$ is known, $\algo{P}$ checks whether $v_j$ is the only predecessor of $w_j$. If yes, $\algo{P}$ predicts ${\sf lab}({v_j})={\sf prelab}(w_j)$ and $\pi({\sf prelab}(v_j))={\sf prelab}(v_j)\xor {\sf lab}({v_j})$. If no, then $\algo{P}$ first finds the other predecessor $u$. If $u\notin  \mbox{\bf Critical}(t_i, t_{i+1})$, as the red-blue pebbling is legal and $w_j$ gets pebbled at round $\gamma>i$, there exists a $\gamma'$ with $i<\gamma'<\gamma$ such that $u$ is pebbled in round $\gamma'$. $\algo{P}$ recognizes the correct call for $u$ and has already updated ${\sf lab}(u)$. If $u\in  \mbox{\bf Critical}(t_i, t_{i+1})$, and the first critical call for $u$ is before round $\gamma$, then ${\sf lab}(u)$ is already known to $\algo{P}$.If $u\in  \mbox{\bf Critical}(t_i, t_{i+1})$ and the first critical call for $u$ is same as $v_j$, since ${\sf lab}({v_j})$ is still unknown, the predictor extract ${\sf lab}(u)$ from the hint $H$. 
\end{itemize}
What is left to do is to describe the working of the predictor for the (later) correct calls for the elements in $\mbox{\bf Critical}(t_i, t_{i+1})$. For each node  $v_j\in \mbox{\bf Critical}(t_i, t_{i+1})$ and each correct ideal primitive call for $v_j$, since the predictor has already computed ${\sf lab}({v_j})$, $\algo{P}$ answers the call without querying the permutation. $\algo{P}$ returns ${\sf lab}({v_j})\xor x$. Note, as ${\sf lab}({v_j})={\sf postlab}(v_j)\xor {\sf prelab}(v_j)$ for a forward call $\pi(x)$ with $x={\sf prelab}(v_j)$, the correct response is indeed ${\sf lab}({v_j})\xor x$. For an inverse call $\pi^{-1}(x)$ with $x={\sf postlab}(v_j)$, the correct response is also ${\sf lab}({v_j})\xor x$. In addition, $\algo{P}$ enters $(v, {\sf prelab}(v), {\sf postlab} (v),{\sf lab}(v))$ in its list.

For each round $\gamma>i$, after checking correct and critical calls for all nodes in $\mbox{\bf Critical}(t_i, t_{i+1})$, the predictor answers the remaining calls by making queries to the random permutation.

\subsubsection{Correctness of the Predictor}
\label{sec:corr-pred}


Recall from Claim~\ref{claim:nocoll}, for a non-colliding input vector $x\in\bool^{n_s n}$, with probability $1-\frac{|V|^2}{2^n}$, all the {\sf prelab}els are distinct. Thus given a list of {\sf prelab}els of all the nodes, we can match the nodes with their {\sf prelab}els. 

If $Bad$ event occurs and $\algo{P}$ is given the hints, it will correctly predict $\pi({\sf prelab}(v))$ for all $v\in \mbox{\bf Critical}(t_i, t_{i+1})$. From Claim~\ref{claim:nocoll}, the set $\{{\sf prelab}(v)\}_{v\in \mbox{\bf Critical}(t_i, t_{i+1})}$ are all distinct and thus contains  at least $(10\delta-1)m$ many points.  \\

First, we argue the correctness of labels maintained by $\algo{P}$ by induction on the order of queries. $\algo{P}$ starts with the correct labels of some nodes in the hint, thus proving the base case. If $v\notin  \mbox{\bf Critical}(t_i, t_{i+1})$, then a correct call for $v$ is preceded by a correct call for $u$ where $u\in{\sf pred}(v)$. Thus by induction hypothesis ${\sf lab}(u)$ is already correctly computed and $\algo{P}$ correctly computes ${\sf prelab}(v)$. Hence $\algo{P}$ can compute ${\sf lab}(v)$ by querying $\pi$ at the first correct call for the node $v$. On the other hand, if $v\in \mbox{\bf Critical}(t_i, t_{i+1})$, then by the input to the query to the successor node and the hint $H$, $\algo{P}$ correctly computes ${\sf lab}(v)$. 
Finally, we show that $\algo{P}$ answers the future correct calls for nodes in the critical set without querying the oracle. For each node $ v\in \mbox{\bf Critical}(t_i, t_{i+1})$, $\algo{P}$ knows the index of the first correct call for $v$ from the hint $L$. It answers the query correctly from the already computed (extracted) and finalized ${\sf lab}(v)$. From the call itself $\algo{P}$ extracts ${\sf prelab}(v)$ and ${\sf postlab}(v)$. For all further correct calls, $\algo{P}$ can respond from its list.    

Finally, we shall use the following lemma, which is an adaptation of Lemma~\ref{lemma:old} in the random permutation setting.

\begin{lemma}
  \label{lemma:newold}
  Let $\pi\sample \Pi$ be a random permutation over $\booln$. Let $\alA$ be an algorithm that takes a hint $h\in \Omega$ as input, makes $q\leq 2^{n-2}$ many queries to $\pi$ and outputs guesses for $k\leq 2^{n-2}$ many un-queried points. The probability that for some $h\in \Omega$, the algorithm $\alA(h)$ successfully guesses all the permutation values correctly is at most $\frac{|\Omega|}{2^{kn-1}}$.
\end{lemma}

Now in our case, total size of the hint is $2^{(10\delta-3)mn}$ and if $Bad$ holds and there is no collision found by $\alA$, then $k\geq (10\delta-1)m$. Hence $Pr[Bad]\leq \frac{2^{(10\delta-3)mn}}{2^{ (10\delta-1)mn-1}}+\frac{|V|^2}{2^n} $. This proves Lemma~\ref{lemma:predictor}. \qed

\subsubsection{Finishing proof of Theorem \ref{thm:smallblock}}
\label{sec:finishing-theorem-}
Fix $\epsilon$ such that $\epsilon-2^{-2mn+1}- 2^{-n(1-1/4\delta)} \geq \epsilon/2$. So far, we have proved the following. For any algorithm $\alA$ that computes the graph function $f_{G,\pi,\tau}$ with probability more than $\epsilon$, it holds with probability more than $\epsilon/2$ that
\begin{itemize}
\item The pre-labels are distinct.
\item The extracted red blue pebbling is legal and has cost at most  $\sum_{i=1}^{k} (20\delta mc_b+\sum_{j\in (t_i,t_{i+1}]} c_r (P_{j}\setminus P_{j-1}))$ where $k$ is the number of interval partitions.
 \item For all intervals $(t_i,t_{i+1}]$ except the last, the algorithm $\alA$ transfers more than $mn$ bits between the cache and the memory and thus the total cost is at least $\sum_{i=1}^k (mc_b+ c_r\sum_{j\in (t_i,t_{i+1}]} |P_j\setminus P_{j-1}|)-mc_b$.  
\end{itemize}

Thus we get that

\begin{align*}
  ecost(f_{G,\pi,\tau},mn) &\geq \epsilon/2 \left( \sum_{i=1}^k (mc_b+ c_r\sum_{j\in (t_i,t_{i+1}]} |P_j\setminus P_{j-1}|)-mc_b\right)\\& \geq \frac{\epsilon}{40\delta}rbcost(G,20\delta m)-\frac{\epsilon mc_b}{2}. 
\end{align*}

This finishes the proof of Theorem \ref{thm:smallblock}. \qed

\section{Wide Block Labelling Functions}
In \cite{C:CheTes19}, authors introduced the notion of \emph{wide block labelling functions}. They instantiated such functions using small block labelling functions and showed wide block labelling functions are useful for construcing succint iMHFs. Startng from a small block labelling  function $\flab^{\ipa}:\bool^L \cup \bool^{2L} \to \bool^L$, \cite{C:CheTes19} constructed a family of wide block labelling functions $\vlab^{\ipa}_{\gamma, W}: \bool^{\gamma W} \to \bool^W$ for $\gamma < \delta$. Such family of labelling function is denoted by $H_{\delta, W}$. They showed if $\flab^{\ipa}$ is pebbling reducible and $G$ is a  source to sink depth robust directed acyclic graph with indegree $\delta$, then the graph function $\mathcal{F}_{G, H_{\delta, W}}$ built upon $G$ and $H_{\delta, W}$ has high $\CMC$ cost. 

We use the same wide block labelling function construction and show if $G$ is a  source to sink depth robust directed acyclic graph with indegree $\delta$, then the graph function $\mathcal{F}_{G, H_{\delta, W}}$ has high \emph{ecost}.

\begin{theorem}
    Suppose, $\flab^{\pi}:\bool^n \cup \bool^{2n} \to \bool^n$ is the small block labelling function defined in Section \ref{sec:graph-based-bandw}\footnote{The labelling function defined in Section \ref{sec:graph-based-bandw} is not dependent on the node itself. Hence it can be viewed as a function whose input is the predcessor nodes.}, $H_{\delta, W}$ is a family of wide block labelling functions $\vlab^{\pi}_{\gamma, W}: \bool^{\gamma W} \to \bool^W$ for $\gamma < \delta$ based on $\flab^{\pi}$ as defined in Section 4.1 of \cite{C:CheTes19}, $G=(V,E)$ be a first-predecessor-distinct $(e,d)$-depth-robust DAG with $\delta$ (maximum) indegree and single source/sink. If $\mathcal{F}_{G, H_{\delta, W}}$ is the graph function built upon $G$ and $H_{\delta, W}$ then $$ecost_\epsilon(\mathcal{F}_{G, H_{\delta, W}}, mn) \geq \frac{\epsilon}{40} \cdot K\sqrt{c_bc_r\delta K e (d-1)} - \frac{5\epsilon}{2} mc_b,$$
    where $K = W/n$, $c_b$ is the per unit memory transfer cost, $c_r$ is the cost of one ideal primitive query and $m$ is the cache size.
\end{theorem}
\begin{proof}
    By opening the graph $\mathcal{F}_{G, H_{\delta, W}}$, one can get the graph $\texttt{Ext}_{\delta, K}(G)$ (as defined in proof of Theorem 4 in section 4.2 of \cite{C:CheTes19}), which is of maximum in degree $2$. By Theorem \ref{thm:smallblock}, we have $$ ecost_\epsilon(\mathcal{F}_{G, H_{\delta, W}}, mn) \geq \frac{\epsilon}{40}rbcost(\texttt{Ext}_{\delta, K}(G),20m)-\frac{\epsilon mc_b}{2}.$$ By Theorem 1.2 of \cite{CCS:BloRenZho18}, we have $$rbcost(\texttt{Ext}_{\delta, K}(G),20m) \geq 2\sqrt{2 c_b c_r \texttt{cc}({\texttt{Ext}_{\delta, K}(G)})} -80mc_b,$$ where $\texttt{cc}(G)$ is cumulative black pebbling complexity of a graph $G$. In proof of Theorem 4 in Section 4.2 of \cite{C:CheTes19} authors also showed $$\texttt{cc}({\texttt{Ext}_{\delta, K}(G)}) \geq \frac{\delta K^3}{8} \cdot e \cdot (d-1).$$ Combining the above three inequalities, we have $$ ecost_\epsilon(\mathcal{F}_{G, H_{\delta, W}}, mn) \geq \frac{\epsilon}{40} \cdot K\sqrt{c_bc_r\delta K e (d-1)} - \frac{5\epsilon}{2} mc_b.$$
\end{proof}
\section{Concluding Discussion.}
 Our main contribution was to show how one can instantiate bandwidth hard functions in the random permutation model, given a constant in-degree graph $G$ with high red-blue pebbling complexity. The result is a natural follow-up of the line of work started with ~\cite{STOC:AlwSer15} followed by \cite{EPRINT:BirDinKho15,CCS:AlwBloHar17,C:CheTes19,CCS:BloRenZho18} and many others.  

In \cite{STOC:AlwSer15}, Alwen and Serbinenko introduced the notion of amortized cost and defined the notion of Memory Hardness in terms of amortized evaluation cost of the function. Intuitively, amortized cost is the right notion for evaulating memory hardness because in reality adversary is usually interested in evaulating the function on a set of inputs rather than a single input. Alwen and Serbinenko also showed, if one time evaluation complexity or Cumulative Memory Complexity (CMC) of the function can be reduced to black pebbling complexity of the underlying graph, then in random oracle model amortized complexity of the function can also be reduced to the same. 

For a family of functions $\mathcal{F}=\{f^\pi\colon\calx\to\caly\}_{\pi\in\Pi}$ we define  $\mathcal{F}^{\otimes n}=\{(f^\pi)^{\otimes n} \colon $ $\calx^n \to \caly^n\}_{\pi\in\Pi}$ where $(f^\pi)^{\otimes n}$ simply extends the domain and range of $f^\pi$ by evaulating the function in $n$ times in parallel. Following \cite{STOC:AlwSer15}, for any $k\in \NN$ one can define amortized cumulative memory complexity (aCMC) as follows:
$$aCMC_{k,\epsilon} \eqdef \min_{\tilde{k} \in [k]}\frac{CMC_{\epsilon}(\calf^{\otimes \tilde{k}})}{\tilde{k}}.$$

Similarly, we can also define amortized energy cost of a functions as follows:
$$aecost_{k, \epsilon}(\calf,mn) \eqdef \min_{\tilde{k} \in [k]}\frac{ecost_{\epsilon}(\calf^{\otimes \tilde{k}},mn)}{\tilde{k}}.$$

While \cite{C:CheTes19,CCS:BloRenZho18} did not  explicitly addressed how to extend their respective results to amortized complexity. However, we note the extensions are not so diffcult. In case of \cite{CCS:BloRenZho18} the extension follows because for node disjoint DAGs the red blue pebbling complexity is additive (same as black pebbling complexity). We observe that for the results in \cite{C:CheTes19}, one can actually obtain a tighter reduction compared to Alwen and Serbinenko's reduction in \cite{STOC:AlwSer15}. This is due to the fact that their ideal primitive does not take any explicit information about the node itself. Similarly, we can also extend our result to amortized energy complexity in random permutation model.

\section{Missing Proofs}
\label{sec:missing-proofs}

\subsection{Proof of Claim~\ref{claim:nocoll}}
\label{sec:proof-claim-nocoll}

  \begin{proof}
    Without loss of generality, we consider the nodes in the topological order. Let $E1_i$ be the event that  ${\sf lab}(u)={\sf lab}(v)$ for some $u\neq v$ with $u,v\leq i$. Let $E2_i$ denote the event ${\sf prelab}(u)={\sf prelab}(v)$ for some $u\neq v$ with $u,v\leq i$. We prove the claim by induction on $i$. First we recall that all the input blocks are distinct as the input is non-colliding.

    For $i=1$, the base case, we start with the observation both $\Pr[E1_1]$ and $\Pr[E2_1]$ is equal to zero as there is only one node with the topological index 1.

 Let $p_1(u),\ldots,p_\delta(u)$ be the parents of $u$ in the topologically sorted order. For any $i\leq m$ the equality ${\sf prelab}(i)={\sf prelab}(m+1)$ holds if and only if
    \begin{align*}
      \bigoplus_{j=1}^\delta {\sf lab} (p_j(i))= \bigoplus_{j=1}^\delta {\sf lab} (p_j(m+1))
    \end{align*}
    Suppose Without loss of generality, the label of $p_\delta(m+1)$ was the last one evaluated. If it was a forward query, then ${\sf postlab} (p_\delta(m+1))$ had $2^n-m$ many choices. Conditioned on $\neg E1_m$, the probability that the above equality gets satisfied is $\frac{1}{2^n-2m}$. Moreover, for an inverse $\pi$ query probability that the output matches with ${\sf prelab} (p_\delta(m+1))$ is also $\frac{1}{2^n-2m}$. Taking union bound over all possible $m$ choices of $i$, we get
    \begin{align*}
      \Pr[E2_{m+1}\mid \neg (E1_m \vee E2_m)]\leq \frac{m}{2^n-2m}
    \end{align*}
    Conditioned on $\neg E2_{m+1}$, the probability that ${\sf lab}(i)= {\sf lab}(m+1)$ for any fixed $i\leq m$ is $\frac{1}{2^n-2m}$. Taking union bound over all $i\leq m$, we get
    \begin{align*}
      \Pr[E1_{m+1}\mid \neg E2_{m+1}\wedge \neg (E1_m \vee E2_m)]\leq \frac{m}{2^n-2m}
    \end{align*}
    Now we bound
    \begin{align*}
       \Pr_\pi\left[{\tt Coll}\right]\leq \sum_{m=1}^{|V|} \frac{2m}{2^n-2m} \leq \frac{2|V|^2}{2^n}
    \end{align*}
    \qed\end{proof}

\subsection{Proof of Lemma~\ref{lemma:size}}
\label{sec:proof-lemma-size}

\begin{proof}
  \textsc{Proof of first statement.}
  By construction of the interval, $|\mbox{\bf Critical}(j, t_{i+1}-1)|\leq (10\delta -1)m$. As cache size is bounded by $m$, $|{\sf pred}(P_{t_{i+1}}\setminus P_{t_{i+1}-1})|\leq m$. Thus  $\lvert\mbox{\bf Critical}(j, t_{i+1})\rvert\leq 10\delta m$. \qed

\textsc{Proof of second statement.}  We start from the observation that ${\sf pred}(P_{j+1}\setminus P_{j})\subseteq \mbox{\bf Critical}(j, t_{i+1})$. For any $v\in \mbox{\bf Critical}(j, t_{i+1})$, either $v\in \mbox{\bf Critical}(t_i, t_{i+1})$ (thus $v\in R_j^{move}$) or $v$ has been pebbled at some step within the interval $(t_i,j)$ (thus $v\in R_j^{old}$). As $R_j=R_j^{move}\cup R_j^{old}$, we conclude ${\sf pred}(P_{j+1}\setminus P_{j})\subseteq R_j$. Hence all the parent nodes are in cache, and hence the pebbling is legal. \qed

\textsc{Proof of third statement.} Follows from $R_j^{old}\subseteq \mbox{\bf Critical}(j+1, t_{i+1}) $ and by Lemma \ref{lemma:size}, $|\mbox{\bf Critical}(j+1, t_{i+1})|\leq 10\delta m$.\qed
\end{proof}
\subsection{Proof of Lemma~\ref{lemma:newold}}
\label{sec:proof-lemma-newold}

\begin{proof}
  Fix an $h$ independent from $\pi$. Given that $\alA$ makes $q$ many queries, the outputs for $q$ many points are fixed. Hence conditioned on the transcript, probability that $\alA$ predicts the permutation value of all the $k$ many points is
  \begin{align*}
    \frac{1}{(2^n-q)(2^n-q-1)\cdots(2^n-q-k+1)}
    &<\frac{1}{2^{kn}-qk2^{n(k-1)}}\\
    &< \frac{1}{2^{kn}-2^{n-1+kn-n}}\\
    &< \frac{1}{2^{kn-1}}
  \end{align*}
  Now taking union bound over all possible choice of $h$, we get the probability is at most $\frac{|\Omega|}{2^{kn-1}}$.\qed 
\end{proof}
\bibliographystyle{plain}
\bibliography{../cryptobib/abbrev3,../cryptobib/crypto,thispaper}

\end{document}